\documentclass[preprint]{iucr}              

                   \def\href#1{\relax}\let\foo\caption
\ifPDF
  \RequirePackage{hyperref}
  \PassOptionsToPackage{pdftex,bookmarksopen,bookmarksnumbered}{hyperref}
\fi
\let\caption\foo

\usepackage{amsthm}
\usepackage{mathptmx}
\usepackage{amssymb}
\usepackage{amsfonts}
\usepackage{tikz}
\usepackage{mathptmx}
\usepackage{amsmath,amscd}
\usepackage{pgfplots}
\usepgflibrary{shapes.geometric} 
\usepgflibrary[shapes.geometric] 
\usetikzlibrary{shapes.geometric} 
\usetikzlibrary{arrows,shapes,backgrounds}

\usepackage{tikz}
\usepackage{pifont}

 \newtheorem{thm}{Theorem}[section]
 \newtheorem{cor}[thm]{Corollary}
 \newtheorem{lem}[thm]{Lemma}
 \newtheorem{prop}[thm]{Proposition}
 \theoremstyle{definition}
 \newtheorem{defn}[thm]{Definition}
 \theoremstyle{remark}

 \numberwithin{equation}{section}

\newcommand{\Spin}{\mathop{\mathrm{Spin}}}
\newcommand{\Cl}{\mathop{\mathrm{Cl}}}
\newcommand{\Aut}{\mathop{\mathrm{Aut}}}

     \paperprodcode{a000000}      
     \paperref{xx9999}            
     \papertype{FA}               

     \paperlang{english}          
     \journalcode{A}             
     \journalyr{2013}
     \journalreceived{\relax}
     \journalaccepted{\relax}
     \journalonline{\relax}

\begin{document}                  



\title{Platonic solids generate their four-dimensional analogues}
\shorttitle{4D polytopes from 3D spinors}


     \cauthor[a,b,c]{Pierre-Philippe}{Dechant}
{ppd22@cantab.net}{}

\aff[a]{Institute for Particle Physics Phenomenology, Ogden Centre for Fundamental Physics, Department of Physics,  University of Durham, South Road, \city{Durham, DH1 3LE}, \country{United Kingdom}}
\aff[b]{Physics Department, Arizona State University, \city{Tempe, AZ 85287-1604}, \country{United States}}
   \aff[c]{Mathematics Department, University of York, \city{Heslington, York, YO10 5GG}, \country{United Kingdom}}


\shortauthor{P-P Dechant}







\maketitle                        

\begin{synopsis}
A Clifford spinor construction shows how the Platonic solids induce their four-dimensional counterparts and determine their symmetries. 
\end{synopsis}

\begin{abstract}
	In this paper, we show how   regular convex 4-polytopes -- the analogues of the Platonic solids in four dimensions --  can be constructed from three-dimensional considerations concerning the Platonic solids alone. 
		Via the Cartan-Dieudonn\'e theorem, the reflective symmetries of the Platonic solids  generate rotations. 
	In a Clifford algebra framework, the space of spinors generating such three-dimensional rotations has a natural four-dimensional Euclidean structure.
The  spinors arising from the Platonic Solids can thus in turn be interpreted as vertices in four-dimensional space, giving a simple construction of the 4D polytopes  16-cell, 24-cell, the $F_4$ root system and the 600-cell. 
	In particular, these polytopes have `mysterious' symmetries, that are almost trivial when seen from the three-dimensional spinorial point of view.
	In fact, all these induced polytopes are also known to be root systems and thus generate rank-4 Coxeter groups, which can be shown to be a general property of the  spinor construction.
	These considerations thus also apply to other root systems such as $A_1 \oplus I_2(n)$ which induces $I_2(n)\oplus I_2(n)$, explaining the existence of the grand antiprism and the snub 24-cell, as well as their symmetries. 
	We discuss these results in the wider mathematical context of Arnol'd's trinities and the McKay correspondence. 
	These results are thus a novel link between the geometries of three and four dimensions, with interesting potential applications on both sides of the correspondence, to real 3D systems with polyhedral symmetries such as (quasi)crystals and viruses, as well as 4D geometries arising  for instance in  Grand Unified Theories and String and M-Theory.
	\\IPPP/13/28, DCPT/13/56\\
\end{abstract}


%
	
\section{Introduction} \label{ACA_Intro}

%

The Platonic solids are the regular convex polytopes in three dimensions; that is they consist of identical vertices and faces that are themselves regular polygons. There are five such solids, namely the cube (8 vertices, 6 faces) and the octahedron (6 vertices, 8 faces),  which are dual under the exchange of face midpoints and vertices, the dual pair dodecahedron  (20 vertices, 12 faces) and icosahedron  (12 vertices, 20 faces), and the self-dual tetrahedron  (4 vertices, 4 faces). These objects are familiar from everyday life, and have in fact been known to humankind for millennia, in particular  at least a thousand years prior to Plato to the neolithic people in Scotland. However, the solids have also always inspired `cosmology', and are named after Plato for their use in his philosophy, in which four of the solids explain the elements (icosahedron as water, cube as earth, octahedron as air and the tetrahedron as fire) and the dodecahedron is the ordering principle of the universe.  Johannes Kepler also attempted to explain the planetary orbits in terms of the Platonic solids, and more recent attempts include the Moon model of the nucleus \cite{Hecht2004Moon} and the Poincar\'e dodecahedral space model of the universe \cite{Luminet2003Dodecahedral}. These more recent fundamental attempts aside, the Platonic solids feature prominently in the natural world wherever geometry and symmetry are important, for instance, in lattices and quasicrystals, molecules such as fullerenes and viruses.  The symmetries of the Platonic solids -- the Coxeter (reflection) groups $A_3$, $B_3$ and $H_3$ for the tetrahedron, cube/octahedron and icosahedron/dodecahedron respectively --  and related Coxeter group symmetries also arise in theoretical physics, for instance in the context of gravitational singularities \cite{HennauxPersson2008SpacelikeSingularitiesAndHiddenSymmetriesofGravity} or the study of topological defects such as the Skyrme model of the nucleus \cite{Manton2004topological}.

The Platonic solids have counterparts in four dimensions. Generalisations of the tetrahedron, cube and octahedron exist in any dimension (the hypersimplex, hypercube and hyperoctahedron), but dimension four is special in that it has three exceptional cases of regular convex polytopes much like the Platonic solids in three dimensions (dodecahedron and icosahedron). These are the hypericosahedron or 600-cell and its dual the 120-cell with symmetries given by the exceptional Coxeter group $H_4$ (which is the largest non-crystallographic Coxeter group and therefore has no higher-dimensional counterpart), and the self-dual 24-cell related to the exceptional phenomena of triality of $D_4$ and the Coxeter group $F_4$. The peculiarities also include mysterious symmetries of these `4D Platonic solids' and the property that several  are root systems (only the octahedron is a root system in 3D), including the hyperoctahedron (or 16-cell) with its dual hypercube, the 8-cell. The 4-simplex is also called the 5-cell, and is self-dual. A summary of regular convex polytopes is displayed in Table \ref{tab_Solids}. 

We therefore adopt the language of Coxeter groups and root systems as appropriate for the description of the reflection symmetry groups of the Platonic solids and their generalisations. Clifford's geometric algebra has an elegant way of handling orthogonal transformations, in particular a very simple description of reflections and rotations. However, an application to the 
root system framework appears only to have been performed in \cite{Dechant2012CoxGA}\cite{Dechant2012AGACSE}. Polytopes in different dimensions are not commonly thought to be related. However, our Clifford/Coxeter approach makes a novel link by showing that the Platonic solids in fact induce their four-dimensional counterparts and their symmetries via a Clifford spinor construction, which explains all the above exceptional, accidental peculiarities of four dimensions.  

 Coxeter groups in dimension four actually feature prominently in high energy physics, and the spinorial nature of our construction could thus have interesting consequences. For instance, $D_4$ is related to the $SO(8)$ symmetry of the transverse dimensions in string theory, and the accidental triality property is crucial for showing the equivalence of the Ramond-Neveu-Schwarz and the Green-Schwarz strings. Similarly $B_4$ corresponds to $SO(9)$ as the little group in M-theory, and $A_4$ is related to $SU(5)$ Grand Unified Theories. All three groups are in turn contained in  the larger exceptional groups $F_4$ and $H_4$, which could themselves become phenomenologically important and their spinorial nature could have interesting consequences. 

Whilst the literature contains partial, loosely connected results on the existence of quaternionic descriptions of these root systems and their automorphism groups (see, e.g. \cite{Humphreys1990Coxeter} and a series of papers by Koca  \cite{Koca2006F4}), we do not think it is a very useful approach and giving a summary would necessarily be very long and fragmented (some more details are contained in \cite{Dechant2012CoxGA}). We believe that we are the first to give a straightforward and uniform proof of their existence and structure. Furthermore, our Clifford spinor approach has the additional benefit of a geometric understanding over a purely algebraic approach, and it is clear what results mean geometrically at any  conceptual stage. This approach thus reveals novel links between the Platonic solids and their four-dimensional counterparts. 

Our link between the Platonic solids, and more generally the spinorial nature of various 4D phenomena could therefore result in a plethora of unknown connections due to  a novel spinorial view of symmetries, for instance in the context of Arnol'd's trinities \cite{Arnold2000AMS} and the McKay correspondence \cite{Mckay1980graphs}. 

The article begins with a review of some necessary background in the Coxeter group and root system framework and in Clifford algebra in Section \ref{ACA_MathBack}.
Section \ref{ACA_Plato} shows how the 3D Platonic solids induce their 4D analogues, and discusses the encountered structures in the context of trinities. 
Section \ref{ACA_general} explains the general nature of the Clifford spinor construction and analyses related 4D polytopes, root systems and symmetry groups. 
Section \ref{ACA_Menagerie} contains a summary of all the rank-4 Coxeter groups in the context of the spinor construction.  
This general aspect of the construction is reminiscent of the McKay correspondence, which we discuss in Section \ref{ACA_McKay} together with the trinities, before we conclude in 
Section \ref{ACA_Concl}.


%
	
\section{Mathematical Background} \label{ACA_MathBack}

In this section, we introduce some simple background in the areas of Coxeter groups, root systems and Clifford algebras, which will be all we need to prove the results in this manuscript.

\subsection{Coxeter Groups} \label{ACA_MathBackCox}

\begin{defn}[Coxeter group] A {Coxeter group} is a group generated by some involutive generators $s_i, s_j \in S$ subject to relations of the form $(s_is_j)^{m_{ij}}=1$ with $m_{ij}=m_{ji}\ge 2$ for $i\ne j$. 
\end{defn}
The  finite Coxeter groups have a geometric representation where the involutions are realised as reflections at hyperplanes through the origin in a Euclidean vector space $\mathcal{E}$ and are thus essentially  just the classical reflection groups. In particular, let $(\cdot,  \cdot)$ denote the inner product in $\mathcal{E}$, and $\lambda$, $\alpha\in\mathcal{E}$.  
\begin{defn}[Reflections and roots] The generator $s_\alpha$ corresponds to the {reflection}
\begin{equation}\label{reflect}
s_\alpha: \lambda\rightarrow s_\alpha(\lambda)=\lambda - 2\frac{(\lambda, \alpha)}{(\alpha, \alpha)}\alpha
\end{equation}
 in a hyperplane perpendicular to the  {root vector} $\alpha$.
\end{defn}

The action of the Coxeter group is  to permute these root vectors, and its  structure is thus encoded in the collection  $\Phi\in \mathcal{E}$ of all such roots, which form a root system: 
\begin{defn}[Root system] \label{DefRootSys}
{Root systems} are defined by the  two axioms
\begin{enumerate}
\item $\Phi$ only contains a root $\alpha$ and its negative, but no other scalar multiples: $\Phi \cap \mathbb{R}\alpha=\{-\alpha, \alpha\}\,\,\,\,\forall\,\, \alpha \in \Phi$. 
\item $\Phi$ is invariant under all reflections corresponding to vectors in $\Phi$: $s_\alpha\Phi=\Phi \,\,\,\forall\,\, \alpha\in\Phi$.
\end{enumerate}
\end{defn}

A subset $\Delta$ of $\Phi$, called the {simple roots}, is sufficient to express every element of $\Phi$ via a $\mathbb{Z}$-linear combination with coefficients of the same sign. $\Phi$ is therefore  completely characterised by this basis of simple roots, which in turn completely characterises the Coxeter group.

Here we are primarily interested in the Coxeter groups of ranks 3 and 4. For the crystallographic root systems, the classification in terms of Dynkin diagrams essentially follows the one familiar from Lie groups and Lie algebras, as their Weyl groups are precisely the crystallographic Coxeter groups. A mild generalisation to so-called Coxeter-Dynkin diagrams is necessary for the non-crystallographic groups:  nodes still correspond to simple roots, orthogonal roots are not connected, roots at $\frac{\pi}{3}$ have a simple link, and other angles $\frac{\pi}{m}$ have a link with a label $m$. For instance, the icosahedral group $H_3$ has one link labelled by $5$, as does its four-dimensional analogue $H_4$, and the infinite two-dimensional family $I_2(n)$ (the symmetry groups of the regular $n$-gons) is labelled by $n$. Table \ref{tab_Corr} displays the groups and their diagrams that are relevant to our discussion. 
Table \ref{tab_Summ} contains a summary of the Platonic solids and their symmetry groups, as well as the root systems of those symmetry groups and a choice for the simple roots. Root systems and their Coxeter groups are classified in the same way (sometimes the `Weyl groups' are also denoted $W(\Phi)$), so that we will move quite freely between them in places.

\subsection{Geometric Algebra} \label{ACA_MathBackGA}

The study of Clifford algebras and Geometric Algebra originated with Grassmann's \cite{Grassmann1844LinealeAusdehnungslehre}, Hamilton's \cite{Hamilton1844} and Clifford's \cite{Clifford1878} geometric work. 
However, the geometric content of the algebras was soon lost when interesting algebraic properties were discovered in mathematics, and Gibbs advocated the use of the hybrid system of vector calculus 
 in physics. When Clifford algebras resurfaced in physics in the context of quantum mechanics, it was purely for their algebraic properties, and this continues in particle physics to this day. Thus, it is widely thought that Clifford algebras are somehow intrinsically quantum mechanical in nature. The original geometric meaning of Clifford algebras has been revived in the work of David Hestenes 
\cite{Hestenes1966STA}, \cite{HestenesSobczyk1984} and \cite{Hestenes1990NewFound}. Here, we follow an exposition along the lines of \cite{LasenbyDoran2003GeometricAlgebra}.

In a manner reminiscent of complex numbers carrying both real and imaginary parts in the same algebraic entity, one can consider the 
geometric product of two vectors defined as the sum of their scalar (inner/symmetric) product and  wedge (outer/ exterior/antisymmetric) product
\begin{equation}\label{in2GP}
    ab:= a\cdot b + a\wedge b.
\end{equation}
The wedge product is the outer product introduced by Grassmann as an antisymmetric product of two vectors, which  naturally defines a plane. Unlike the constituent inner and outer products, the geometric product is invertible, as $a^{-1}$ is simply given by $a^{-1}=a/(a^2)$. This leads to many algebraic simplifications over standard vector space techniques, and also feeds through to the differential structure of the theory, with Green's function methods that are not achievable with vector calculus methods.
This geometric product can be extended to the product of more vectors via associativity and distributivity, resulting in higher grade objects called multivectors.   There are a total of $2^n$ elements in the algebra, since it truncates at grade $n$ multivectors due to the scalar nature of the product of  parallel vectors and the antisymmetry of orthogonal vectors. Essentially, a Clifford algebra is a deformation of the exterior algebra by a quadratic form, and for a Geometric Algebra this is the metric of space(time).

The geometric product provides a very compact and efficient way of handling reflections in any number of dimensions, and thus by the Cartan-Dieudonn\'e theorem also rotations.  For a unit vector $n$, we consider the reflection of a vector $a$ in the hyperplane orthogonal to $n$. Thanks to the geometric product, in Clifford algebra the two terms in Eq. (\ref{reflect}) combine into a single term, and thus a `sandwiching prescription': 
\begin{thm}[Reflections]\label{HGA_refl}
In Geometric Algebra, a vector `$a$' transforms under a reflection in the (hyper-)plane defined by a unit normal vector `$n$' as
	\begin{equation}\label{in2refl}
	  a'=-nan.
	\end{equation}
\end{thm}

This is a remarkably compact and simple prescription for reflecting vectors in hyperplanes. More generally, higher grade multivectors of the form $M= ab\dots c$ (so-called versors) transform similarly (`covariantly'), as $M= ab\dots c\rightarrow \pm nannbn\dots ncn=\pm nab\dots cn=\pm nMn$. Even more importantly, from the  Cartan-Dieudonn\'e theorem, rotations are the product of successive reflections. For instance, compounding the reflections in the hyperplanes defined by the unit vectors $n$ and $m$ results in a rotation in the plane defined by $n\wedge m$.
\begin{prop}[Rotations]\label{HGA_rot}
In Geometric Algebra, a vector `$a$' transforms under a rotation in the plane defined by $n\wedge m$ via successive reflection in hyperplanes determined by the unit vectors `$n$' and `$m$' as 
	\begin{equation}\label{in2rot}
	  a''=mnanm=: Ra\tilde{R},
	\end{equation}
where we have defined $R=mn$ and the tilde denotes the reversal of the order of the constituent vectors $\tilde{R}=nm$.
\end{prop}

\begin{thm}[Rotors and spinors]\label{HGA_thm_Rotor}
The object $R=mn$ generating the rotation in Eq. (\ref{in2rot}) is called a rotor. It satisfies $\tilde{R}R=R\tilde{R}=1$. Rotors themselves transform single-sidedly under further rotations, and thus form a multiplicative group under the geometric product, called the rotor group. Since $R$ and $-R$ encode the same rotation, the rotor group is a double-cover of the special orthogonal group, and is thus essentially the Spin group. Objects in Geometric Algebra that transform single-sidedly are called spinors, so that rotors are  normalised spinors. 
\end{thm}


Higher multivectors transform in the above covariant, double-sided way as $ MN\rightarrow (RM\tilde{R})(R N \tilde{R})=RM\tilde{R}R N \tilde{R}=R(MN)\tilde{R}$.

	The Geometric Algebra of three dimensions $\Cl(3)$ spanned by three orthogonal (thus anticommuting) unit vectors $e_1$, $e_2$ and $e_3$ contains three bivectors $e_1e_2$, $e_2e_3$ and $e_3e_1$ that square to $-1$, as well as the  highest grade object $e_1e_2e_3$   (trivector and pseudoscalar), which also squares to $-1$.
	\begin{equation}\label{in2PA}
	  \underbrace{\{1\}}_{\text{1 scalar}} \,\,\ \,\,\,\underbrace{\{e_1, e_2, e_3\}}_{\text{3 vectors}} \,\,\, \,\,\, \underbrace{\{e_1e_2=Ie_3, e_2e_3=Ie_1, e_3e_1=Ie_2\}}_{\text{3 bivectors}} \,\,\, \,\,\, \underbrace{\{I\equiv e_1e_2e_3\}}_{\text{1 trivector}}.
	\end{equation}

\begin{thm}[Quaternions and spinors of $\Cl(3)$]\label{HGA_quatBV}
The unit spinors $\lbrace 1,-Ie_1, -Ie_2, -Ie_3\rbrace$ of $\Cl(3)$ are isomorphic to the quaternion algebra $\mathbb{H}$. 
\end{thm}

Most of the results we will derive in this manuscript are therefore readily translated into the language of quaternions. However, we will refrain from doing so at every step and instead advocate the geometric approach in terms of spinors. This offers a new coherent picture, from which the plethora of loosely connected results without geometric insight from the literature follows in a straightforward and uniform way.

\subsection{3D root systems induce 4D root systems} \label{ACA_MathBackInd}

The following is a summary of Ref.  \cite{Dechant2012Induction} which proves that every root system in three dimensions induces a root system in four dimensions in completely general terms, using only the Coxeter and Clifford frameworks, but making no reference to any specific root system. The remainder of this manuscript in turn considers the implications of this general statement for the concrete list of root systems in three and four dimensions, including novel links between Arnol'd's trinities and with the McKay correspondence, as well as explaining for the first time the otherwise mysterious structure of the automorphism groups of these root systems. 

The argument in this section is that each root system in 3D allows one to find an even discrete spinor group from the Coxeter reflection root vectors via the geometric product. Because of the spinors' $O(4)$-structure, this spinor group can be reinterpreted as a set of 4D vectors, for which one can then show  the root system axioms to hold.

\begin{prop}[$O(4)$-structure of spinors]\label{HGA_O4}
The space of $\Cl(3)$-spinors can be endowed with an inner product and a norm giving it a 4D Euclidean signature. For two spinors $R_1$ and $R_2$, this is given by 	$(R_1,R_2)=\frac{1}{2}(R_1\tilde{R}_2+R_2\tilde{R}_1)$. 
\end{prop}
\begin{proof}
For a spinor $R=a_0+a_1Ie_1+a_2Ie_2+a_3Ie_3$, this gives $(R,R)=R\tilde{R}=a_0^2+a_1^2+a_2^2+a_3^2$, as required. 
	
\end{proof}

\begin{cor}[3D spinors and 4D vectors]\label{HGA_spinvec}
A spinor in three dimensions induces a vector in four dimensions by mapping the spinor components into the 4D Euclidean space as just defined in Proposition \ref{HGA_O4}. A discrete spinor group thus gives rise to a set of vertex vectors that can be interpreted as a 4D polytope.
\end{cor}

This is in fact already enough for most of our results about the four-dimensional counterparts of the Platonic solids, including their construction and symmetries. However, it is interesting that one can in fact also show the stronger statement that these polytopes have to be root systems and therefore induce Coxeter groups of rank 4.

\begin{lem}[Reflections in 4D]\label{HGA_4Drefl}
A reflection of the vector in the 4-dimensional space corresponding to the spinor $R_2$ under the norm in Proposition \ref{HGA_O4} in the vector corresponding to $R_1$ is given by $R_2\rightarrow R_2'=-R_1\tilde{R}_2R_1/(R_1\tilde{R}_1)$.
\end{lem}
\begin{proof}
	For spinors $R_1$ and $R_2$, the reflection formula (\ref{reflect}) gives
	$R_2\rightarrow R_2'=R_2-2(R_1, R_2)/(R_1, {R}_1) R_1 =R_2-((R_1\tilde{R}_2+R_2\tilde{R}_1) R_1/(R_1\tilde{R}_1)=  -R_1\tilde{R}_2R_1/(R_1\tilde{R}_1)$.
\end{proof}

	In fact, we are mostly interested in unit spinors, for which this simplifies to $-R_1\tilde{R}_2R_1$. It is easily verified in terms of components that this is indeed the same as the usual reflection of 4-dimensional vectors. 

\begin{thm}[Induced root systems in 4D]\label{HGA_4Drootsys}
A 3D root system gives rise to an even spinor group which induces a root system in 4D.
\end{thm}
\begin{proof}
Check the two axioms for root systems for $\Phi$ given by the set of 4D vectors induced by a spinor group.
\begin{enumerate}
\item By construction, $\Phi$ contains the negative of a root since if $R$ is in a spinor group $G$, then so is $-R$ (c.f. Theorem \ref{HGA_thm_Rotor}), but no other scalar multiples. 
\item $\Phi$ is invariant under all reflections given by Lemma \ref{HGA_4Drefl}	since $R_2'=-R_1\tilde{R}_2R_1/(R_1\tilde{R}_1)\in G$ if $R_1, R_2 \in G$ by the closure property of the group $G$ (in particular $\tilde{R}$ is in $G$ if $R$ is). 
\end{enumerate}

\end{proof}

The spinorial nature of these induced root systems is thus critical for the understanding of the closure property -- in particular, it is immediately obvious why $|\Phi|=|G|$ -- and we shall see later that it is also crucial for the analysis of the automorphism groups of these polytopes.

\section{Platonic Relationships} \label{ACA_Plato}

We now turn to concrete examples of three-dimensional root systems and consider which four-dimensional polytopes they induce.

\subsection{The Platonic Solids, Reflection Groups and Root Systems}

We start with the symmetry groups of the Platonic solids $A_3$ (tetrahedron), $B_3$ (octahedron and cube) and $H_3$ (icosahedron and dodecahedron).   The induced polytopes are the 24-cell, which generates the Coxeter group $D_4$ from $A_3$, the root system of $F_4$ from $B_3$, and the 600-cell (the root system of $H_4$) from $H_3$. 
The group $A_1\times A_1 \times A_1$ is also a symmetry of the tetrahedron, which is found to induce the 16-cell, which is the root system of $A_1\times A_1 \times A_1\times A_1$.

The three simple roots of the Coxeter groups are in fact sufficient to generate the entire root systems. The root vectors encoding reflections are then combined to give spinors, as by Cartan-Dieudonn\'e a rotation is an even number of reflections. 

\begin{thm}[Reflections/Coxeter groups and polyhedra/root systems]
Take the three simple roots for the Coxeter group $A_1\times A_1\times A_1$ (respectively $A_3$/$B_3$/$H_3$). Geometric Algebra reflections in the hyperplanes orthogonal to these vectors via Eq. (\ref{in2refl}) generate further vectors pointing to the 6 (resp. 12/18/30) vertices of an octahedron (resp. cuboctahedron/cuboctahedron with an  octahedron/icosidodecahedron), giving the full root system of the group.
\end{thm}

For instance, the simple roots for $A_1\times A_1\times A_1$ are $\alpha_1=e_1$, $\alpha_2=e_2$ and $\alpha_3=e_3$ for orthonormal basis vectors $e_i$. Reflections amongst those then also generate $-e_1$, $-e_2$ and $-e_3$, which all together point to the vertices of an octahedron.

\begin{thm}[Spinors from reflections]\label{HGA_rotors}
The 6 (resp. 12/18/30)  reflections in the  Coxeter group  $A_1\times A_1 \times A_1$ (resp. $A_3$/$B_3$/$H_3$)   generate 8 (resp. 24/48/120) different rotors via Proposition \ref{HGA_rot}.
\end{thm}

For the $A_1\times A_1\times A_1$ example above, the spinors thus generated are $\pm 1$, $\pm e_1e_2$, $\pm e_2 e_3$ and $\pm e_3e_1$.

\begin{thm}[4D polytopes]\label{HGA_4Dvertices}
The set of 8 (resp. 24/48/120) rotors when reinterpreted as a 4D polytope generate the 16-cell (24-cell/24-cell with dual/600-cell) 
\end{thm}

 For the rotors from $A_1\times A_1\times A_1$  one gets the vertices of the 16-cell ($(\pm 1, 0, 0,0)$ and permutations) via the correspondence in Corollary \ref{HGA_spinvec}. 

This is enough for the construction of the counterparts of the Platonic solids in 4D. However, the stronger statement on root systems implies also the following. 

\begin{thm}[4D root systems]\label{HGA_4Dtrin}
The Coxeter group  $A_1\times A_1 \times A_1$ (resp. $A_3$/$B_3$/$H_3$)   generates the root system for $A_1\times A_1 \times A_1\times A_1$ (resp. $D_4$/$F_4$/$H_4$).
\end{thm}

In fact, these groups of discrete spinors yield a novel construction of the binary polyhedral groups.

\begin{thm}[Spinor groups and binary polyhedral groups]
The discrete spinor group in Theorem \ref{HGA_rotors} is isomorphic to the quaternion group $Q$ (resp. binary tetrahedral group $2T$/binary octahedral group $2O$/binary icosahedral group $2I$).
\end{thm}

The calculations are straightforward once the Clifford algebra framework with the geometric product is adopted, and more details can be found in \cite{Dechant2012CoxGA},  \cite{Dechant2012AGACSE}.

The Platonic solids thus in the above sense induce their counterparts in four dimensions, the convex regular polychora. There are six such polytopes, and the 16-cell, 24-cell and 600-cell are directly induced as shown above and displayed in Table \ref{tab_spin_Summ}. Using duality, the 8-cell is induced from the 16-cell and the 120-cell is the  dual of the 600-cell (the 24-cell is self-dual). The only remaining case is the 5-cell. This is the 4-simplex belonging to the family of $n$-dimensional simplices with symmetry group $A_n$. This is the only such 4D polytope that is not equal or dual to a root system. In fact it can obviously not be a root system, nor in particular be constructed via our approach, as it has an odd number of vertices, 5. This is therefore (ironically) the only exception to our connections among the Platonic solids and their four-dimensional counterparts. The only regular polytopes in higher dimensions are the $n$-dimensional simplex ($A_n$), cube ($B_n$) and crosspolytope ($B_n$). Thus, in particular the existence of the exceptional 4D phenomena of 24-cell ($D_4$  and $F_4$), 600-cell and 120-cell ($H_4$) are explained by the `accidentalness' of the spinor construction. This is particularly interesting for triality ($D_4$), $F_4$ as the largest crystallographic group in 4D, and quasicrystals, since $H_4$ is the largest non-crystallographic Coxeter group.

\subsection{Arnol'd and Mathematical Trinities}\label{ACA_Trin}

The great mathematician Vladimir Arnol'd had an exceedingly broad view of mathematics, and his metapattern-inspired proofs and conjectures have started and/or shaped many subject areas  \cite{Arnold2000AMS}. For instance, linear algebra is essentially the theory of the root systems $A_n$. However by abstracting away towards a description in terms  of root systems, many results carry over to other root systems and thereby to other geometries (e.g. Euclidean and symplectic for $BC_n$, $D_n$). This is an alternative to the conventional view of seeing these as special cases of linear algebra with extra structure. 

The most recent and  important such metapattern appear to be his trinities \cite{Arnold2000AMS} \cite{Arnold1999symplectization}, born out of the observation that many areas of real mathematics can be complexified and quaternionified resulting in theories with a similar structure. The fundamental trinity is thus $(\mathbb{R}, \mathbb{C}, \mathbb{H})$, and other trinities include $(\mathbb{R}P^n, \mathbb{C}P^n, \mathbb{H}P^n)$, $(\mathbb{R}P^1=S^1, \mathbb{C}P^2=S^2, \mathbb{H}P^1=S^4)$, the M\"obius/Hopf bundles $(S^1\rightarrow S^1, S^4\rightarrow S^2, S^7 \rightarrow S^4)$,  $(E_6, E_7, E_8)$ and many more. 

There are in fact trinities related to the above Platonic considerations such as (Tetrahedron, Octahedron, Icosahedron), $(A_3, B_3, H_3)$, $(24, 48, 120)$, and $(D_4, F_4, H_4)$ but they were very loosely connected to each other in previous work. For instance, Arnol'd's connection between $(A_3, B_3, H_3)$ and $(D_4, F_4, H_4)$ is very convoluted and involves numerous other trinities at intermediate steps via a decomposition of the projective plane into Weyl chambers and Springer cones, and noticing that the number of Weyl chambers in each segment $24=2(1+3+3+5), 48=2(1+5+7+11), 120=2(1+11+19+29))$ miraculously matches the quasihomogeneous weights $((2, 4, 4, 6), (2, 6, 8, 12), (2, 12, 20, 30))$ of the Coxeter groups $(D_4, F_4, H_4)$ \cite{Arnold1999symplectization}. 

We therefore believe that the construction here is considerably easier and more immediate than Arnol'd's original connection between several of the trinities, such as   $(A_3, B_3, H_3)$,   $(D_4, F_4, H_4)$, (Tetrahedron, Octahedron, Icosahedron), and $(24, 48, 120)$. In fact we are not aware that the following are considered trinities and would suggest to add them: the root systems of $(A_3, B_3, H_3)$ (cuboctahedron, cuboctahedron with octahedron, icosidodecahedron), the number of roots in these root systems $(12, 18, 30)$ and the binary polyhedral groups $(2T, 2O, 2I)$. 

Our framework also finds alternative interpretations of well-known trinities, such as $(24, 48, 120)$ as the number of 3D spinors or 4D root vectors as opposed to the Weyl number decomposition. We will revisit these connections and interpretations in more detail later in the context of the McKay correspondence, as one can wonder if this picture in terms of trinities is in fact the most useful description. For instance, the Clifford spinor construction also worked for $A_1\times A_1 \times A_1$ giving the 4D `Platonic solid' 16-cell, and we shall see in the next section that the construction also holds for the other 3-dimensional root systems, arguably making it more general than a trinity. 


Going back to the beginning of this section, the spinorial nature of the root systems $(D_4, F_4, H_4)$ could also have interesting consequences from the perspective of abstracting away from linear algebra to $A_n$ and generalising to other root systems and geometries.

\section{The general picture: 3D root systems, spinor induction  and symmetries}\label{ACA_general}

The Clifford spinor construction holds for any rank-3 root system, and not just those related to the Platonic solids as considered above. In this section we therefore examine the remaining cases. In fact, the root systems $A_3$, $B_3$ and $H_3$ are the only irreducible root systems in 3 dimensions. $A_1$ is the unique one-dimensional root system, and having already considered  $A_1\times A_1\times A_1$, the only missing cases are the sum of $A_1$ with a two-dimensional irreducible root system. These are the root systems of the symmetry groups $I_2(n)$ of the regular $n$-gons, which are easily dealt with in a uniform way.

\subsection{A Doubling Procedure}\label{ACA_Double}

Without loss of generality,  the simple roots for $I_2(n)$  can  be taken as
 $\alpha_1=e_1$ and $\alpha_2=-\cos{\frac{\pi}{n}}e_1+\sin{\frac{\pi}{n}}e_2$. We have shown in \cite{Dechant2012Induction} that an analogue of the spinor construction exists in 2D, but is of limited interest, as the 2D root systems are shown to be self-dual: 

The space of spinors $R=a_1+a_2e_1e_2=: a_1+a_2I$ in two-dimensional Euclidean space (defining $I:=e_1e_2$) is also two-dimensional, and has a natural Euclidean structure given by $R\tilde{R}=a_1^2+a_2^2$. A 2D root vector $\alpha_i=a_1e_1+a_2e_2$ is therefore in bijection with a spinor by $\alpha_i\rightarrow \alpha_1\alpha_i=e_1\alpha_i =a_1+a_2e_1e_2= a_1+a_2I$ (taking $\alpha_1=e_1$ without loss of generality). This is the same as forming a spinor between those two root vectors.  The infinite family of two-dimensional root systems $I_2(n)$ is therefore self-dual. 

Taking $\alpha_1$ and $\alpha_2$ as generating $I_2(n)$ and $\alpha_3=e_3$ for $A_1$, one has a total of $2n+2$ roots. One easily computes that these generate  a spinor group of order $4n$  which consists of two sets of order $2n$ that are mutually orthogonal under the spinor norm in Proposition \ref{HGA_O4}. One therefore finds the following theorem. 

\begin{thm}[4D root systems from $A_1\oplus I_2(n)$]\label{HGA_4DI2n}
Under the Clifford spinor construction the 3D root systems  $A_1\oplus I_2(n)$    generate the root systems $ I_2(n)\oplus I_2(n)$ in 4D.
\end{thm}

The case of $A_1\times A_1\times A_1$ inducing $A_1\times A_1\times A_1\times A_1$ is now seen to be a special case of this more general `doubling construction'. In fact one can easily see that one of the $I_2(n)$ sets is $e_1e_3$-times that of the other. In terms of quaternions (Theorem \ref{HGA_quatBV}), this corresponds to an imaginary unit $j$ or $k$ and is often the starting Ansatz in the literature \cite{Koca2009grand}. This is in fact the only way in 4D the root systems can be orthogonal, but we just point out here that it arises naturally from our induction construction. 

To see why the order of the spinor group is $4n$ and the construction yields two copies with the above properties, let us consider the  products of two root vectors. If both root vectors in the product $\alpha_i\alpha_j$ are from $A_1$, one merely gets $\pm 1$, which is trivially in the spinor group. Without loss of generality one can therefore say that either one or both root vectors are from $I_2(n)$ (there are $2n$ root vectors). If both are from $I_2(n)$, then from the self-duality of $I_2(n)$ one has that $2n$ such spinors $R=\alpha_i\alpha_j$ arise. It is easy to see that none of these can contain  $e_3$.
The other possibility is to have one root $\alpha_i$ from $I_2(n)$ and $\alpha_3$ from $A_1$. There are $2n$ of the former and because they contain the negative roots $-\alpha_i$, only $2n$ different spinors arise when multiplying with $\pm \alpha_3$. 
These therefore together account for the order of $4n$. Since the first case of spinor is in bijection with a root vector via multiplying with $e_1$, one can continue and map to the second case by multiplying with $e_3$. One can therefore map directly from one kind of spinor to the other by multiplying with $e_1e_3\sim j$. The two are therefore necessarily orthogonal but otherwise identical.

\subsection{Spinorial Symmetries}\label{ACA_Sym}

The Clifford algebraic approach via spinor groups has the decided advantage that it is clear firstly why the root system is given by a binary polyhedral group, and secondly, why this group reappears in the automorphism group. There are three common group actions (whereby the group acts on itself): left action (action by group multiplication from the left) $gh$, right action $hg$, and conjugation $g^{-1}hg$. By virtue of being a spinor group, the set of vertex vectors is firstly closed under reflections  and thus a root system because of the group closure property, and secondly invariant under both left and right multiplication separately.

\begin{thm}[Spinorial symmetries]\label{HGAsymmetry}
A root system induced via the Clifford spinor construction has an automorphism group that contains two factors of the respective spinor group acting from the left and the right.
\end{thm}

In our opinion, the construction from 3D is the only compelling explanation for a number of features that we will explain for the example of $H_4$:
\begin{itemize}
	\item That the root system $H_4$ can be constructed in terms of quaternions (Theorems \ref{HGA_quatBV} and \ref{HGA_O4}).
	\item That reflections are given by quaternion multiplication (Lemma \ref{HGA_4Drefl}).
	\item That as a discrete quaternion group the root system is isomorphic to a discrete subgroup of $\Spin(3)\sim SU(2)$, the binary icosahedral group (Theorem \ref{HGA_thm_Rotor}).
	\item That the group $H_4$ can essentially be generated from 2 (rather than 4) simple quaternionic roots (they are essentially the spinors $\alpha_1\alpha_2$ and $\alpha_2\alpha_3$ in terms of the simple roots of $H_3$ \cite{Dechant2012CoxGA}).
	\item That the sub root system $H_3$ is given by the pure quaternions (this is essentially just Hodge duality with the pseudoscalar/inversion $I$, mapping root vectors to pure quaternions. For instance, this is not true for $A_3$, which does not contain $I$ \cite{Dechant2012CoxGA}).
	\item That the automorphism group of $H_4$ consists of two copies of  the binary icosahedral group $2I$ (Theorem \ref{HGAsymmetry}): $\Aut(H_4)=2I\times 2I$ and is of order $(120)^2$.
    \item That $H_4$ is an exceptional phenomenon (accidentalness of the construction).
\end{itemize}
Similarly, the automorphism group of $F_4$ is given by the product of two binary octahedral groups $\Aut(F_4)=2O\times 2O$ of order $(48)^2$. The automorphism group of $D_4$ contains two factors of the binary tetrahedral group $2T$ of order $(24)^2$, as well as an order 2 $\mathbb{Z}_2$-factor, which is essentially whether the basis vectors $e_1$, $e_2$, $e_3$ are cyclic or anticyclic (i.e. a Dynkin diagram symmetry of $A_3$). In particular, $D_4$ does not contain $A_3$ as a pure quaternion subgroup, since $A_3$ does not contain the inversion, and the central node in the $D_4$ diagram is essentially spinorial (i.e. not a pure bivector/quaternion). The automorphism groups of $I_2(n)\oplus I_2(n)$ are two factors of the dicyclic groups of order $(4n)^2$. The automophism group of $A_1^4$ contains two copies of the quaternion group $Q$ as well as a factor of $S_3$ for the permutations of the basis vectors $e_1$, $e_2$, $e_3$ (from the Dynkin diagram symmetries of $A_1^3$), giving order $3!8^2$. A summary of the symmetries in this and the next subsections is displayed in Table \ref{tab_Symm}. 

We therefore contend that the Clifford algebraic approach in terms of spinors is a new geometric picture which derives the known results (and more) uniformly and much more efficiently than the standard approach. In particular, spinor techniques extend to arbitrary dimensions -- the isomorphism with $\mathbb{H}$ is accidental in three dimensions and the quaternionic description does therefore not extend to higher dimensions. 

The spinorial nature of the respective 4D root systems thus demystifies the peculiar symmetries of the 4D Platonic solid analogues 16-cell, 24-cell (and dual) and 600-cell and their duals as essentially the rotational symmetries of the conventional 3D Platonic solids. In the next two subsections, we consider another group action, having dealt with left and right actions in this section, and we shall see that the spinorial symmetries also leave imprints on other 4D  (semi-regular) polytopes.

\subsection{Conjugal Spinor Groups}\label{ACA_Conj}

The spinor groups we have been considering so far formed groups where the multiplication law was given by the geometric product. However, as we have seen, these groups are also closed when one takes the operation $R_2\rightarrow R_2'=-R_1\tilde{R}_2R_1/(R_1\tilde{R}_1)$ from Lemma \ref{HGA_4Drefl} as the group multiplication. This of course simply amounts to closure under reflections in 4D and thus the root system property. 

If one considers the spinor groups derived from reflections in $A_3$/$B_3$/$H_3$ (i.e. essentially $2T$/$2O$/$2I$ which  ultimately gives rise to $D_4$/$F_4$/$H_4$) as given earlier in Theorem \ref{HGA_rotors}, but now instead takes as the group multiplication law the one given by Lemma \ref{HGA_4Drefl}, one finds  several subgroups, which of course correspond to the sub root systems that one would expect such as $A_1^n$, $A_3$, $B_3$, $H_3$, $A_2$, $H_2$, $A_2\times A_1$, $A_2\times A_1\times A_1$ etc and their closure property. However, one also finds in addition $A_2\times A_2$ or $H_2\times H_2$ in the case of $H_4$, or $B_4$ in $F_4$.

Since the double-sided multiplication law is remotely reminescent of group conjugation (with a twist), as opposed to left and right action which gave rise to the automorphism groups, we will call these subgroups `conjugal'. It is interesting that it is possible to recast the problem of finding a sub root system to the group theoretic problem of finding subgroups.

\subsection{The Grand antiprism and the snub 24-cell}\label{ACA_GrAp}

Since the $H_4$ root system 600-cell contains $H_2\oplus H_2$, it is obvious that $H_2\oplus H_2$ even when thought of as a subset of $H_4$ is invariant under its own Coxeter group, so that it lies on its own orbit. The 600-cell has  120 vertices given by the binary icosahedral group $2I$, and one finds that subtracting the 20 vertices of $H_2\oplus H_2$, the remaining 100 vertices are on another orbit of $H_2\oplus H_2$, and give a semi-regular polytope called the Grand antiprism. It was only constructed in 1965 by Conway and Guy by means of a computer calculation \cite{Conway1965Four}. In particular it is interesting that the symmetry group of the grand antiprism is by our construction  given by $\Aut(H_2\oplus H_2)$ \cite{Koca2009grand}, which as we have just seen is of order $400=20^2$. This route to the grand antiprism is considerably more economical than the traditional approach. It is interesting as a non-Platonic example of a spinorial symmetry and also from the doubling perspective: $H_3$ has a subgraph $H_2\oplus A_1$ (by ignoring the unlabelled link in $H_3$), so one might think of the  $H_2\oplus H_2$ inside $H_4$ as induced via the doubling procedure from the $H_2\oplus A_1$ inside the  $H_3$. Likewise, $H_4$ also has another (maximal) subgroup  $\Aut(A_2\oplus A_2)$ that can similarly be seen to arise from the $A_2 \oplus A_1$ inside the $H_3$ by deleting the other link (the one labelled by $5$) in the $H_3$ diagram, and has order $144=12^2$. These are intriguing imprints of spinorial geometry on the symmetries of the Grand antiprism. 

The snub 24-cell has similar symmetries. The binary tetrahedral group $2T$ is a subgroup of the binary icosahedral group $2I$. Therefore, subtracting the 24 vertices of the 24-cell from the 120 vertices of the 600-cell, one gets a semi-regular polytope with 96 vertices called the snub 24-cell. Since the 24 subtracted points from the $D_4$ root system form a single orbit  under $2T$, the remaining 96 points are likewise separately left invariant under $2T$. The symmetry group of both sets is therefore given by $2T\times 2T$, and the order is thus $576=24^2$, explaining the symmetry of the snub 24-cell in spinorial terms. 
 
These two cases of 4D polytopes are therefore examples of semiregular polytopes exhibiting spinorial symmetries, much like the `4D Platonic solids'.

\section{The 4D Menagerie}\label{ACA_Menagerie}

We have shown that some Coxeter groups of rank 4 are induced via the Clifford spinor construction, and we have seen that others are subgroups or conjugal subgroups of these. There is only a limited number of rank 4 root systems. Therefore in this section we consider all rank 4 root systems in the context of spinor induction. 

Table \ref{tab_rank4} summarises the results for the menagerie of 4D root systems. The first column in the table denotes the decomposition of the rank in terms of the rank of the irreducible components. In particular, all irreducible rank-4 Coxeter groups are either spinor induced (denoted by a tick $\checkmark$) or (conjugal) subgroups of those that are (denoted by  $\sim$). Likewise all Coxeter groups that are the product of a rank-3 group with $A_1$ can be obtained as (conjugal) subgroups of the irreducible ones. For the $2+2$ decomposition we encounter the case $I_2(n)\oplus I_2(n)$ that we found was induced from $I_2(n)\oplus A_1$. However, the general case $I_2(n)\oplus I_2(m)$ is the first case that cannot in general be spinorially induced (denoted by   $\text{\sffamily X}$). Likewise, $I_2(n)\times A_1\times A_1$ is neither spinor induced, nor a subgroup of the larger Coxeter groups for general $n$. However, the special case of $A_1^4$ was our first example of   spinor induction. 

In general, one would not expect most, and certainly not all, such rank-4 Coxeter groups to be spinor induced, as the series $A_n$, $B_n$ and $D_n$ exist in any dimension and one can form sums from smaller irreducible components. However, it is striking how many of them are inducible via spinors, in particular all those associated with exceptional phenomena in four dimensions. At this point, we therefore go back to our considerations of exceptional phenomena, trinities and the McKay correspondence.

\section{Arnol'd's Trinities and the McKay correspondence}\label{ACA_McKay}

In this section, we discuss a wider framework with multiple connections amongst trinities and different interpretations for them. However, we have also seen that the Clifford spinor construction is more general, and perhaps more akin to the McKay correspondence, than a trinity. We therefore begin by introducing the McKay correspondence.

The trinities $(2T, 2O, 2I)$ and $(E_6, E_7, E_8)$ of the binary polyhedral groups and the $E$-type Lie groups are connected via the McKay correspondence in the following sense. The binary polyhedral groups are discrete subgroups of $SU(2)$ and therefore each have a 2-dimensional irreducible spinor representation $2_s$. We can define a graph by assigning a node to each irreducible representation of the binary polyhedral groups with the following rule for connecting edges: each node corresponding to a certain irreducible representation is connected to the nodes corresponding to those irreducible representations that are contained in its tensor product with $2_s$. For instance, tensoring the trivial representation $1$ with $2_s$ trivially gives $2_s$ and thus the only link  $1$ has is with $2_s$; $2_s\otimes 2_s=1+3$, such that $2_s$ is connected to $1$ and $3$, etc. 
 On the Lie group side one considers the affine extension of $(E_6, E_7, E_8)$ achieved by extending the graph of the Dynkin diagram by an extra node. The McKay correspondence is the observation that the graphs derived in both ways are the same, as shown in Figure \ref{figE8aff}. In particular the affine node on the Lie group side corresponds to the trivial representation  of the binary polyhedral groups. There are other mysterious connections, for instance the coefficients of the highest/affine root  of the affine Lie group in terms of the roots of the unextended Lie group are given by the dimensionalities of the irreducible representations of the corresponding binary group. 

However, the McKay correspondence is more general than this relation between trinities, for it  holds for all finite subgroups of $SU(2)$, in particular the ones that have the 2-dimensional discrete subgroups of $SO(3)$ as preimages  under the universal covering map. This way the infinite families of the cyclic groups and the dicyclic groups correspond to the infinite families of affine Lie groups of $A$- and $D$-type. The McKay correspondence is therefore more a result on the $ADE$-classification than a mere trinity. In the sense that our Clifford spinor construction also applies to the infinite family of 2-dimensional groups $I_2(n)$, it feels closer in spirit to the McKay correspondence. 

In fact there is now an intricate web of connections between trinities, some well-known and several we believe to be new, as well as trinities appearing in different guises in multiple interpretations, as shown in Figure \ref{figMcKay}. 

The Clifford spinor construction inducing $(2T, 2O, 2I)$ from $(A_3, B_3, H_3)$ does not seem to be known, and the $(2T, 2O, 2I)$ then  induce the root systems $(D_4, F_4, H_4)$. The $(2T, 2O, 2I)$ also correspond to $(E_6, E_7, E_8)$ via the McKay correspondence. The affine Lie groups have the same Coxeter-Dynkin diagram symmetries as 
$(D_4, F_4, H_4)$, i.e. $S_3$ (triality) for $D_4$ and $E_6^+$, $S_2$  for $F_4$ and $E_7^+$, and $S_1$  for $H_4$ and $E_8^+$ ($H_4$ and $E_8$ also have the same Coxeter number/element, as easily shown in Clifford algebra), making a connection between these two trinities. 

In fact, $(A_3, B_3, H_3)$, $(2T, 2O, 2I)$ and $(E_6, E_7, E_8)$ are connected in one chain via $(12, 18, 30)$, which we have not encountered in the literature and which we suggest as a trinity in its own right. $(12, 18, 30)$ are the Coxeter numbers of $(E_6, E_7, E_8)$ -- performing all 6/7/8 fundamental reflections in the Coxeter groups $(E_6, E_7, E_8)$ corresponding to the simple roots  gives the so-called Coxeter elements $w$ of the groups; their order $h$ ($w^h=1$) is called the Coxeter number. However, $(12, 18, 30)$ is also the sum of the dimensions of the irreducible representations ($\sum d_i$) of the binary polyhedral groups. It does not appear to be known that this is also connected all the way to $(A_3, B_3, H_3)$, as $(12, 18, 30)$ is also the number of roots in their root systems.  

Similarly there is a chain linking $(A_3, B_3, H_3)$, $(2T, 2O, 2I)$ and $(D_4, F_4, H_4)$ via the trinity $(24, 48, 120)$. It is at the same time the number of different spinors generated by the reflections in $(A_3, B_3, H_3)$, the order of the binary polyhedral groups $(2T, 2O, 2I)$ given by the sum of the squares of the dimensions of the irreducible representations ($\sum d_i^2$), the number of roots of the 4-dimensional root systems $(D_4, F_4, H_4)$, as well as the square root of the order of their automorphism group. 

Without doubt, there are more connections to be found, and deeper reasons for these connections to exist, so we propose here the Clifford algebra approach as a novel and hopefully fruitful path to explore.

\section{Conclusions}\label{ACA_Concl}

In the literature, great significance is attached to quaternionic representations, in particular those in terms of pure quaternions. We have shown that this belief is misplaced, and that the situation is much clearer and more efficiently analysed in a geometric setup in terms of spinors. The pure quaternion sub root systems are not in fact deeply mysterious yet significant subsets of the rank-4 groups (something that only works if the group contains the inversion), but the rank-4 groups are instead induced from 3-dimensional considerations alone, and do not in fact contain more geometric content than that of three dimensions alone. Or perhaps the mystery  is  resolved and  the significance explained, now that there is a simple geometric explanation for it.

We have found novel connections between the Platonic solids and their four-dimensional counterparts, as well as other 4D polytopes.
In particular, our construction sheds light on the existence of all the exceptional phenomena in 4D such as self-duality of the 24-cell and triality of $D_4$, the exceptional root systems $F_4$ (largest crystallographic in 4D) and $H_4$ (largest non-crystallographic). 
The striking symmetries of these four-dimensional polytopes had been noticed but had not really been understood in any geometrically meaningful way. We have made novel connections in pure mathematics over a broad range of topics, and in relation to trinities and the McKay correspondence. The spinorial nature of the rank-4 root systems could also have profound consequences in high energy physics since these groups are pivotal in Grand Unified Theories and String and M Theory. So perhaps after the failed attempts of Plato, Kepler and Moon to order the elements and the universe, planets and nuclei in terms of the Platonic solids, they might still leave their mark on the universe in guises yet to be discovered. 


	\ack{I would like to thank my family and friends for their support and   David Hestenes, Eckhard Hitzer, Anthony Lasenby, Joan Lasenby, Reidun Twarock,  C\'eline B\oe hm, Richard Clawson, and Mike Hobson for helpful discussions, as well as David Hestenes and the ASU Physics department for their hospitality during the final stages of writing up.}





		%
		\begin{table}
		\begin{centering}\begin{tabular}{|c|c|c|c|c|c|}
		\hline
	
		\hline
		3D&Tetrahedron&&Octahedron&Icosahedron
		\tabularnewline
		Dual &self-dual&&Cube&Dodecahedron
		\tabularnewline
		\hline	\hline
		4D&5-cell&24-cell&16-cell&600-cell
		\tabularnewline
		Dual&self-dual&self-dual&8-cell&120-cell
		\tabularnewline
			\hline	\hline
			$n$D&$n$-simplex&&$n$-hyperoctahedron&
			\tabularnewline
			Dual&self-dual&&$n$-hypercube&
			\tabularnewline
		\hline
		\end{tabular}\par\end{centering}
		\caption
		{\label{tab_Solids} The regular convex polytopes in three (Platonic solids), four and higher dimensions (for a discussion of `Platonic Solids' in arbitrary dimensions see, for instance, \cite{Szajewska2012Faces}).  }
		\end{table}
		%


	%
	\begin{table}
	\begin{centering}\begin{tabular}{|c|c||c||c|c||c|c|c|}
	\hline
	rank-3 group&diagram&binary&rank-4 group&diagram&Lie algebra&diagram
	\tabularnewline
	\hline
	\hline
	$A_1\times A_1\times A_1$&{\includegraphics[width=1.5cm]{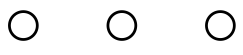}}&$Q$&$A_1\times A_1 \times A_1\times A_1$&{\includegraphics[width=2cm]{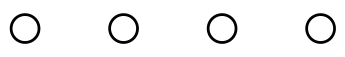}}&$D_4^+$&		{\includegraphics[width=1.5cm]{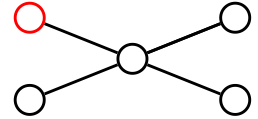}}
	\tabularnewline
	\hline
	$A_3$&{\includegraphics[width=1.5cm]{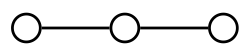}}&$2T$&$D_4$&		{\includegraphics[width=1.5cm]{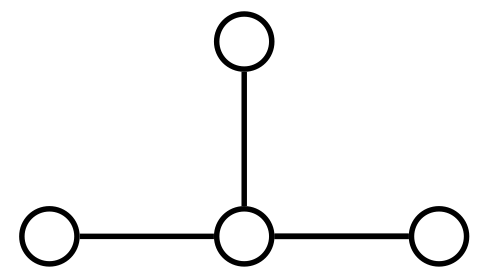}}
	&	$E_6^+$&	{\includegraphics[width=3cm]{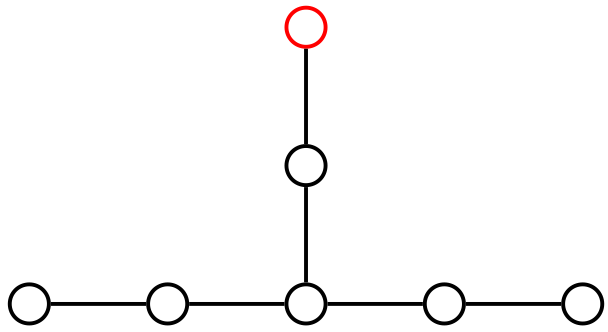}}
	\tabularnewline
	\hline
	$B_3$&{\includegraphics[width=1.5cm]{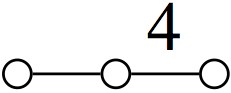}}
	&$2O$&$F_4$&
	{\includegraphics[width=2cm]{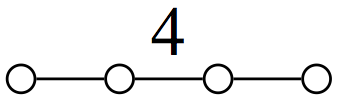}}&	$E_7^+$&	{\includegraphics[width=3.5cm]{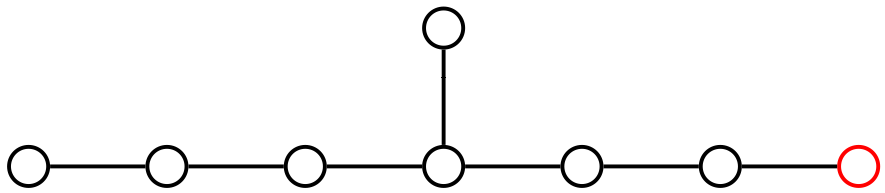}}
	\tabularnewline
	\hline
	$H_3$&{\includegraphics[width=1.5cm]{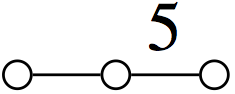}}&$2I$&$H_4$&		{\includegraphics[width=2cm]{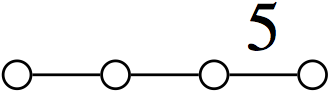}}
	& 	$E_8^+$&	{\includegraphics[width=4cm]{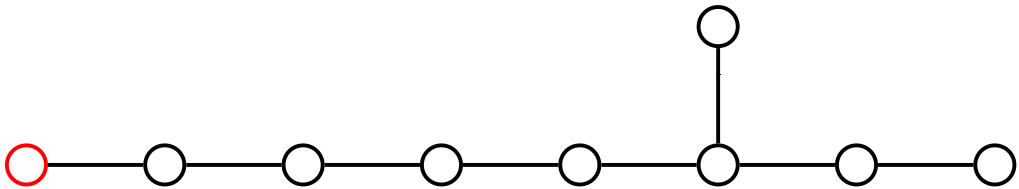}}
	\tabularnewline
	\hline
	\end{tabular}\par

	\end{centering}
	\caption{\label{tab_Corr} Overview over  the Coxeter groups discussed here and their Coxeter-Dynkin diagrams. Correspondence between the rank-3 and rank-4 Coxeter groups as well as the affine Lie algebras (the affine root is in red). The spinors generated from the reflections contained in the respective rank-3 Coxeter group via the geometric product are realisations of the binary polyhedral groups $Q$, $2T$, $2O$ and $2I$, which themselves generate (mostly exceptional) rank-4 groups, and are related to (mostly the type-$E$) affine Lie algebras via the McKay correspondence. }
	\end{table}
	%

	%
	\begin{table}
	\begin{centering}\begin{tabular}{|c|c|c|c|}
	\hline
	Platonic solid&Coxeter group&Root system&Simple roots $\alpha_i$
	\tabularnewline

	\hline
	\hline
	Tetrahedron&$A_1^3$&Octahedron&$(1,0,0)$, $(0,1,0)$, $(0,0,1)$
	\tabularnewline
	&$A_3$&Cuboctahedron&$(1,1,0)$, $(0,-1,1)$, $(-1,1,0)$
	\tabularnewline
	\hline
	Octahedron&$B_3$&Cuboctahedron&$(1,-1,0)$, $(0,1,-1)$, $(0,0,1)$
	\tabularnewline
	Cube&&+Octahedron&
	\tabularnewline
	\hline
	Icosahedron&$H_3$&Icosidodecahedron&$(0,-1,0)$, $(-\sigma,1,\tau)$, $(0,0,-1)$
	\tabularnewline
	Dodecahedron&&&
	\tabularnewline

	\hline
	\end{tabular}\par\end{centering}
	\caption
	{\label{tab_Summ} The reflective symmetries of the Platonic solids. The columns show respectively the Platonic solids,  their reflection symmetry groups (Coxeter groups), their root systems, and a set of simple roots (the normalisation has been omitted for better legibility). Here, $\tau$ is the golden ratio $\tau=\frac{1}{2}(1+\sqrt{5})$, and $\sigma$ is the other solution (its `Galois conjugate') to the quadratic equation $x^2=x+1$, namely $\sigma=\frac{1}{2}(1-\sqrt{5})$.}
	\end{table}
	%


%
\begin{table}
\begin{centering}\begin{tabular}{|c|c|c|c|c|}
\hline
Platonic solid&3D group&Spinors&4D Polytope&4D group
\tabularnewline

\hline
\hline
Tetrahedron&$A_1^3$&$Q$&16-cell&$A_1^4$
\tabularnewline
&$A_3$&$2T$&24-cell&$D_4$
\tabularnewline
\hline
Octahedron&$B_3$&$2O$&$F_4$-root system&$F_4$
\tabularnewline
Cube&&&&
\tabularnewline
\hline
Icosahedron&$H_3$&$2I$&600-cell&$H_4$
\tabularnewline
Dodecahedron&&&120-cell&
\tabularnewline

\hline
\end{tabular}\par\end{centering}
\caption{\label{tab_spin_Summ} Spinors generated by the reflective symmetries of the Platonic solids: the Coxeter reflections generate discrete spinor groups that are isomorphic to the quaternion group $Q$ (or the 8 Lipschitz units, in terms of quaternions), the binary tetrahedral group $2T$ (24 Hurwitz units), the binary octahedral group $2O$ (24 Hurwitz units and their 24  duals) and the binary icosahedral group $2I$ (120 Icosians). These generate certain rank-4 Coxeter groups. When re-interpreting 3D spinors as 4D vectors, these point to the vertices of certain regular convex 4-polytopes.

}
\end{table}
%



%
\begin{table}
\begin{centering}\begin{tabular}{|c|c|c|c|c|c|}
\hline
rank 3&$|\Phi|$&$|\Aut(\Phi)|$&rank 4&$|\Phi|$& $|\Aut(\Phi)|$
\tabularnewline

\hline
\hline
$A_3$&$12$&$24$&$D_4$&$24$&$2\cdot 24^2=1152$
\tabularnewline
\hline
$B_3$&$18$&$48$&$F_4$&$48$&$48^2=2304$ 
\tabularnewline
\hline
$H_3$&$30$&$120$&$H_4$&$120$&$120^2=14400$ 
\tabularnewline
\hline
$A_1^3$&$6$&$8$&$A_1^4$&$8$&$3!\cdot8^2=384$ 
\tabularnewline
\hline
$A_1\oplus A_2$&$8$&$12$&$A_2\oplus A_2$&$12$&$12^2=144$

\tabularnewline
\hline
$A_1\oplus H_2$&$12$&$20$&$H_2\oplus H_2$&$20$&$20^2=400$ 
\tabularnewline
\hline
$A_1\oplus I_2(n)$&$2n+2$&$4n$&$I_2(n)\oplus I_2(n)$&$4n$&$(4n)^2$ 
\tabularnewline
\hline
\end{tabular}\par\end{centering}
\caption{\label{tab_Symm} Summary of the non-trivial symmetries of 4D root systems that can be interpreted as induced from a 3D spinorial point of view: the 24-cell and snub 24-cell; $\Aut(\Phi_{F_4})$; 120-cell and 600-cell; 16-cell;  $\Aut(A_2\oplus A_2)$; the grand antiprism and $\Aut(H_2\oplus H_2)$. More generally, $\Aut(I_2(n)\oplus I_2(n))$ is of order $4n\times 4n$.}
\end{table}
%


%
\begin{table}
\begin{centering}\begin{tabular}{|c|c|c|c|c|c|}

\hline
$4$&$A_4$&$B_4$&$D_4$&$F_4$&$H_4$
\tabularnewline
\hline
&$\sim$&$\sim$&$\checkmark$&$\checkmark$&$\checkmark$
\tabularnewline
\hline
\hline
$3+1$&$A_3\times  A_1$&$B_3\times  A_1$&$H_3\times  A_1$&&
\tabularnewline
\hline
&$\sim$&$\sim$&$\sim$&&
\tabularnewline
\hline
\hline
$2+2$&$I_2(n)\times I_2(n)$&$I_2(n)\times I_2(m)$&&&
\tabularnewline
\hline
&$\checkmark$&$\text{\sffamily X}$&&&
\tabularnewline
\hline
\hline
$2+1+1$&$I_2(n)\times A_1\times A_1$&&&&
\tabularnewline
\hline
&$\text{\sffamily X}$&&&&
\tabularnewline
\hline
\hline
$1+1+1+1$&$A_1^4$&&&&
\tabularnewline
\hline
&$\checkmark$&&&&
\tabularnewline
\hline
\hline
\end{tabular}\par\end{centering}
\caption{\label{tab_rank4} The rank-4 Coxeter groups in terms of irreducible components. A tick $\checkmark$ denotes that the rank-4 group is induced directly via the Clifford spinor construction. A $\sim$ denotes that the group is a subgroup or even has a root system that is the sub root system (i.e. a conjugal subgroup) of a group that is spinor induced. $\text{\sffamily X}$ means that the root system cannot be induced spinorially.}
\end{table}
%


\begin{figure}
	\begin{center}
{\includegraphics[width=8cm]{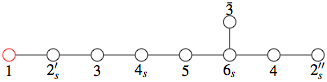}}

\end{center}
\caption{McKay correspondence for $E_8$ and $2I$: Dynkin diagram for the standard affine extension of $E_8$, here denoted $E_8^+$, and the graph for the tensor product structure of the binary icosahedral group $2I$, where nodes correspond to irreducible representations (labelled by their dimension $d_i$, and a subscript $s$ denotes a spinorial representation). The affine root $\alpha_0$ of $E_8^+$ (red) corresponds to the trivial representation of $2I$ and is given in terms of the other roots as $-\alpha_0=\sum d_i \alpha_i$. The sum of the dimensions of the irreducible representations of $2I$ gives the Coxeter number   $\sum {d_i}=30=h$ of $E_8$, and the sum of their squares  $\sum {d_i^2}=120$ gives the order of $2I$.}
\label{figE8aff}
\end{figure}


\begin{figure}
	\begin{center}
{\includegraphics[width=15cm]{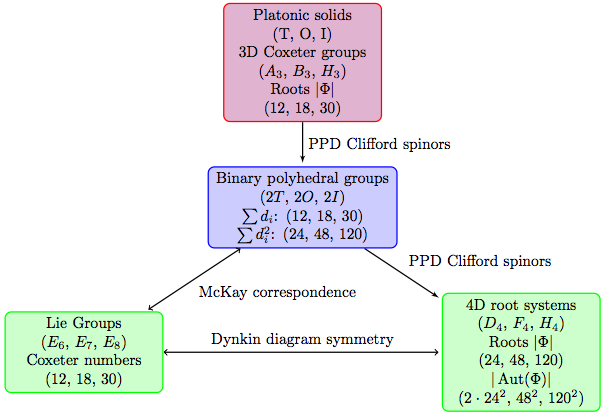}}

\end{center}
\caption{Overview over the web of connections between different trinities via the McKay correspondence and the Clifford spinor construction. }
\label{figMcKay}
\end{figure}


\end{document}